\documentclass[preliminary,copyright,creativecommons,noderivs,noncommercial]{eptcs}
\DeclareTextSymbolDefault{\textbackslash}{T1}

\usepackage{svn-multi}
\usepackage{amsmath}
\usepackage{amssymb}
\usepackage{alltt}
\usepackage{amsthm}
\usepackage{graphicx}
\usepackage{xspace,stmaryrd}
\usepackage{datetime}
\usepackage{etoolbox}
\usepackage{longtable}
\usepackage{ulem}
\usepackage{color}
\usepackage[all,2cell,cmtip,dvips,ps]{xy}

\DeclareSymbolFont{letters}{OML}{txmi}{m}{it}

\newcommand{\tb}{\textbackslash}

\DeclareSymbolFont{operators}{OT1}{cmr}{m}{n}
\DeclareSymbolFont{letters}{OML}{cmm}{m}{it}
\DeclareSymbolFont{symbols}{OMS}{cmsy}{m}{n}
\DeclareSymbolFont{largesymbols}{OMX}{cmex}{m}{n}

\newcommand{\MetaSOS}{{\textsf{Meta~SOS}}}
\newcommand{\PREG}{{\textsf{PREG~Axiomatizer}}}

\newcommand{\codefont}{\fontsize{10}{10}\selectfont}
\newcommand{\code}[1]{{\tt \codefont {#1}}}

\settimeformat{hhmmsstime}
\newcommand{\builddate}{\ifnumcomp{\year}{<}{10}{0}{}\the\year-\ifnumcomp{\month}{<}{10}{0}{}\the\month-\ifnumcomp{\day}{<}{10}{0}{}\the\day\ \currenttime}
\ifnumcomp{\svnhour}{=}{23}{\def\svnhour{01}}{}
\ifnumcomp{\svnhour}{=}{22}{\def\svnhour{00}}{}
\ifnumcomp{\svnhour}{=}{21}{\def\svnhour{23}}{}
\ifnumcomp{\svnhour}{=}{20}{\def\svnhour{22}}{}
\ifnumcomp{\svnhour}{=}{19}{\def\svnhour{21}}{}
\ifnumcomp{\svnhour}{=}{18}{\def\svnhour{20}}{}
\ifnumcomp{\svnhour}{=}{17}{\def\svnhour{19}}{}
\ifnumcomp{\svnhour}{=}{16}{\def\svnhour{18}}{}
\ifnumcomp{\svnhour}{=}{15}{\def\svnhour{17}}{}
\ifnumcomp{\svnhour}{=}{14}{\def\svnhour{16}}{}
\ifnumcomp{\svnhour}{=}{13}{\def\svnhour{15}}{}
\ifnumcomp{\svnhour}{=}{12}{\def\svnhour{14}}{}
\ifnumcomp{\svnhour}{=}{11}{\def\svnhour{13}}{}
\ifnumcomp{\svnhour}{=}{10}{\def\svnhour{12}}{}
\ifnumcomp{\svnhour}{=}{09}{\def\svnhour{11}}{}
\ifnumcomp{\svnhour}{=}{08}{\def\svnhour{10}}{}
\ifnumcomp{\svnhour}{=}{07}{\def\svnhour{09}}{}
\ifnumcomp{\svnhour}{=}{06}{\def\svnhour{08}}{}
\ifnumcomp{\svnhour}{=}{05}{\def\svnhour{07}}{}
\ifnumcomp{\svnhour}{=}{04}{\def\svnhour{06}}{}
\ifnumcomp{\svnhour}{=}{03}{\def\svnhour{05}}{}
\ifnumcomp{\svnhour}{=}{02}{\def\svnhour{04}}{}
\ifnumcomp{\svnhour}{=}{01}{\def\svnhour{03}}{}
\ifnumcomp{\svnhour}{=}{00}{\def\svnhour{02}}{}

\newtheorem{defn}{Definition}

\newtheorem{thm}[defn]{Theorem}

\newcommand{\commform}{\textsf{comm-form}}

\newcommand{\COMM}{\id{COMM}}
\newcommand{\simcom}{\mathrel{\sim_{cc}}}

\newcommand{\sosrule}[2]{\frac{\raisebox{.7ex}{\normalsize{$#1$}}}
                        {\raisebox{-1.0ex}{\normalsize{$#2$}}}}
\newcommand{\trans}[1]{\,{\xrightarrow{{#1}}}\,}
\newcommand{\ntrans}[1]{\,{\stackrel{{#1}}{\nrightarrow}}\,}

\newcommand{\bisim}{\;\underline{\hspace*{-0.15ex}
                        \leftrightarrow\hspace*{-0.15ex}}\;}

\newcommand{\bisimpar}[1]{\;\underline{\hspace*{-0.15ex}
                        \leftrightarrow\hspace*{-0.15ex}}^{{#1}}}

\newcommand{\nul}{{\ensuremath{\mathbf{0}}}}

\newcommand{\DedRule}[1]{\mbox{{\bf\small (#1)}}}

\newcommand{\id}[1]{\mathit{#1}}

\def\lparal{\mathbin{\setbox0=\hbox{$\|$}%
        \dimen0=\dp0 \advance\dimen0 -1.5pt \dp0=\dimen0%
        \underline{\kern-1.5pt\box0\kern1.5pt}}}

\def\lparal{\mathbin{\setbox0=\hbox{$\|$}%
        \dimen0=\dp0 \advance\dimen0 -1.5pt \dp0=\dimen0%
                \underline{\kern-1.5pt\box0\kern1.5pt}}}

\newcommand{\bbn}{|\mkern-2mu[}
\newlength{\chisymbolwidth}
\definecolor{lightblue}{RGB}{224,224,255}
\definecolor{lightred}{RGB}{255,224,224}
\definecolor{lightgreen}{RGB}{224,255,224}
\definecolor{lightyellow}{RGB}{255,255,224}
\definecolor{lightpurple}{RGB}{255,224,255}
\definecolor{darkerred}{RGB}{64,0,0}
\definecolor{darkred}{RGB}{128,0,0}
\definecolor{darkblue}{RGB}{0,0,128}
\definecolor{darkgreen}{RGB}{0,128,0}
\definecolor{darkpurple}{RGB}{128,0,128}

\makeatletter
\def\THICKhrulefill{\leavevmode \leaders \hrule height 5pt\hfill \kern \z@}
\makeatother

\begin{document}

\title{Meta SOS -- A Maude Based SOS Meta-Theory Framework\thanks{The authors have been partially supported by the
 project `Meta-theory of Algebraic Process Theories' (nr.~100014021)
 of the Icelandic Research Fund.  Eugen-Ioan Goriac is also funded by
 the project `Extending and Axiomatizing Structural Operational
 Semantics: Theory and Tools' (nr.~1102940061) of the Icelandic
 Research Fund. The authors want to thank Alberto~Lluch~Lafuente for his comments on an earlier version of the paper. }}

\author {
Luca Aceto${^{1,2}}$ \qquad
Eugen-Ioan Goriac${^{1}}$ \qquad
Anna Ingolfsdottir${^{1}}$
\email{
[luca,egoriac10,annai]@ru.is
}
\institute {
${^{1}}$ ~~ ICE-TCS, School of Computer Science, Reykjavik University, Iceland\\
${^{2}}$ ~~ SysMA, IMT Lucca Institute for Advanced Studies, Lucca 55100, Italy
}
}

\def\authorrunning{Aceto, Goriac \& Ingolfsdottir}
\def\titlerunning{{\MetaSOS} -- A Maude Based SOS Meta-Theory Framework}

\maketitle

\begin{abstract}
{\MetaSOS} is a software framework designed to integrate the results from the meta-theory of structural operational semantics (SOS).
These results include deriving semantic properties of language constructs just by syntactically analyzing their rule-based definition, as well as automatically deriving sound and ground-complete axiomatizations for languages, when considering a notion of behavioural equivalence.
This paper describes the {\MetaSOS} framework by blending aspects from the meta-theory of SOS, details on their implementation in Maude, and running examples.
\end{abstract}

\section{Introduction}\label{Sect:intro}

Structural Operational Semantics \cite{Plotkin81} is a well known approach for intuitively specifying the semantics of programming and specification languages by means of rules.
These rules can be analyzed using meta-theoretic results in order to infer certain properties about language constructs by purely syntactic means.
Research on SOS meta-theory has at its core the development of rule formats that, if respected, will guarantee that some language constructs have certain properties, such as commutativity, associativity, and idempotence.
We refer the reader to \cite{Mousavi09-BEATCS} for an overview on how to derive these properties as well as axiomatizations.
Rule formats can also be used to obtain congruence properties for behavioural equivalences (see, e.g., \cite{Aceto01}) and semantic properties such as determinism of transition relations \cite{AcetoBIMR12}.

Despite the large body of research on the meta-theory of SOS, to the best of our knowledge, there currently does not exist an extensible software tool integrating the results obtained so far in that research area. (We briefly review some of the existing software tools below.) This is an unsatisfactory state of affairs since such a software framework would allow language designers to benefit from the results in the meta-theory of SOS while experimenting with their language designs. The design of programming and specification  languages is a highly non-trivial endeavour and tool support is needed in order to support prototyping of language designs, their algorithmic analysis and early checking of desired semantic properties. The meta-theory of SOS provides, for example,  syntactic criteria guaranteeing the validity of semantic properties, but checking such criteria by hand is error prone and quickly becomes infeasible.

In this paper we introduce {\MetaSOS}{\footnote{The framework is downloadable from \url{http://goriac.info/tools/meta-sos/}.}}, a framework for handling SOS specifications, with the purpose of performing simulations, deriving axiomatizations, and checking for rule formats.
Though it has a different line of implementation, {\MetaSOS} continues the work we started with a prototype named {\PREG} \cite{AcetoCGI11a}, dedicated to deriving axiom systems from SOS specifications.

We are aware of other software tools that are somewhat related to {\MetaSOS}. In \cite{sos-maude} the authors show how to prototype SOS meta-theory in Maude \cite{DBLP:conf/maude/2007}.
That paper was a good point of reference for us both for implementation details and future work ideas.
The Process Algebra Manipulator (PAM) {\cite{DBLP:journals/fmsd/Lin95}} is designed to perform algebraic reasonings on programs written only in CCS \cite{Milner89a}, CSP \cite{Hoare85} and LOTOS \cite{Brinksma-1985a}.
PAM does not allow the user to define their own language.
The Maude MSOS Tool (MMT) \cite{MaudeMSOS} does provide this facility, however it does not focus on axiomatizations or rule formats, and, unfortunately, neither does it facilitate
a natural extension with new features. LETOS \cite{so55686} is a lightweight tool to aid the development of operational semantics, which supports execution and tracing, as well as quality rendering using \LaTeX. LETOS makes some first steps towards checking operational conservativity along the lines proposed in the paper \cite{GrooteV92}.

The rest of the paper is organized as follows.
In Section~\ref{sec:pre} we present some preliminaries on SOS, Maude and {\MetaSOS}. Section~{\ref{sec:components}} describes the three components the framework currently provides: a simulator and bisimilarity checker (Section~\ref{sec:simbis}), a sound and ground-complete axiom schema deriver (Section~\ref{sec:axiomderiver}), and a commutativity format checker (Section~\ref{sec:formatchecker}).
It also includes Section~\ref{sec:linda}, where we present a case study that integrates all the previously mentioned components, and Section~\ref{sec:addingcomponents}, where we briefly show how to extend the framework with more functionalities.
Finally, Section~\ref{sec:conc} concludes the paper and points out possible directions for future research.

\section{Preliminaries}\label{sec:pre}

Maude \cite{DBLP:conf/maude/2007} is a high-level language providing support for specifying multi-sorted signatures, equational and rewrite theories.
Not only is it an excellent environment to perform reasonings with these theories at object level, but also, due to its reflective capabilities, to analyze and operate with them at meta-level.
Previous efforts \cite{AcetoCGI11a, MaudeMSOS, sos-maude, SerbanutaRM09, Verdejo2006226} have shown its suitability in facilitating SOS specifications.

{\MetaSOS} is implemented in Maude as a metalanguage application \cite{GoriacCLAG09}. This means that the framework extends Maude with capabilities such as providing SOS specifications and operating with them.
After opening the Maude environment and loading the framework by using the command \code{load metasos.maude}, a specification is given using the standard syntax for inputting functional modules: \code{(fmod SPECIFICATION is~{\ldots}~endfm)}, where ``\code{\ldots}'' consists of constructs that are discussed in the remainder of the paper.

We assume a \emph{signature} $\Sigma$, which is a set of function symbols with fixed arities (typical members: $f$, $g$). Function symbols with arity 0 are referred to as \emph{constants}.
Moreover, we assume an infinite set of \emph{variables} $V$ (typical members: $x, y$).

\emph{Open terms} are inductively built using variables and function symbols by respecting their arities. The set of open terms is denoted by ${\mathbb{T}}(\Sigma)$ (typical members: $s, t$). By $T(\Sigma)$ we denote the set of terms formed without variables, referred to as \emph{closed terms} (typical members: $p,q$).
\emph{Substitutions}, which are functions of the type $\sigma:V \to {\mathbb{T}}(\Sigma)$, have the role of replacing variables in an open term with other (possibly open) terms.

{\MetaSOS} has a basic set of sorts. One of them represents the domain of process terms ${\mathbb{T}}(\Sigma)$ and has the name \code{PTerm}.
It is important to note that we did not use the name \code{Term} due to it being reserved for operating at meta-level with general terms formed using Maude multi-sorted signatures.
In order to have access to the sort \code{PTerm} one needs to include a core {\MetaSOS} module named \code{RULES} in the specification: \code{including~RULES~.}
Operations are given using a standard syntax.
For instance, the following construct declares a binary operation $f$ over process terms: \code{op f~:~PTerm PTerm -> PTerm [metadata~"sos"]~.}
Notice the use of the attribute in square brackets, which makes it possible for $f$ to be used in SOS specifications. Variables are also given using a standard syntax: \code{var~x~y~:~PTerm~.}

\subsection{Transition System Specifications in {\MetaSOS}}
\label{subsec:tssmetasos}

We will now describe how transition system specifications are expressed in {\MetaSOS}.

\begin{defn}[Transition System Specification]
Consider a signature $\Sigma$ and a set of labels $L$ (with typical members $l, l'$), $t, t' \in {\mathbb{T}}(\Sigma)$ and $l \in L$.
A \emph{positive transition formula} is a triple $(t,l,t')$, written $t \trans{l} t'$, with the intended meaning: process $t$ performs the action labelled as $l$ and becomes process $t'$. A \emph{negative transition formula} is a tuple $(t,l)$, written $t \ntrans{l}$, with the meaning that process $t$ cannot perform the action labelled as $l$.

A \emph{transition rule} is a pair $(H,\alpha)$, where $H$ is a set of formulae and $\alpha$ is a formula. The formulae from $H$ are called \emph{premises} and the formula $\alpha$ is called the \emph{conclusion}.
A transition rule is often denoted by $\sosrule{H}{\alpha}$ and has the following generic shape:
\[
\sosrule{\{t_i \trans{l_{i}} t'_{i}  \mid  i \in I\} \cup \{t_{j} \ntrans{l_{j}} \mid j \in J \}}{t \trans{l} t'}
,\]
where $I, J$ are index sets, $t, t',t_{i}, t'_{i}, t_{j} \in {\mathbb{T}}(\Sigma)$, and $l_{i}, l_{j} \in L$.
A \emph{transition system specification} (abbreviated TSS) is a triple $(\Sigma, L,
{\cal R})$ where $\Sigma$ is a signature, $L$ is a set of labels, and ${\cal R}$ is a set of transition rules of the provided shape.

\end{defn}

In {\MetaSOS}, positive and negative formulae are denoted by expressions such as \code{t~-(l)->~t'} and  \code{t~-(l)/>}, respectively. Here \code{t}, \code{t'} are variables of sort \code{PTerm} and \code{l} is a variable of another provided sort, \code{PLabel}.
A transition rule $\sosrule{H}{c}$ is declared as \code{H === c}, where \code{H} consists of a (possibly empty) list of comma-separated formulae.
The entire set of transition rules is given as a list of rules wrapped in a Maude membership axiom declaration: \code{mb~\ldots~:~PAllRules~.}

\newcommand{\BCCSP}{{\textnormal{BCCSP}}}

To exemplify a full {\MetaSOS} specification, consider the {\BCCSP} system from \cite{Glabbeek01}. Its signature $\Sigma_{\BCCSP}$ includes the \emph{deadlock} process $\nul$, a collection of \emph{prefix} operators $l.\_$ ($l \in L$) and the binary \emph{choice} operator $\_+\!\_$\,. For a fixed $L = \{a,b,c\}$, the deduction rules for these operators are:
\[
\sosrule{}{l.x \trans{l} x} \hspace{3ex} \sosrule{x \trans{l} x'}{x + y \trans{l} x'} \hspace{3ex} \sosrule{y \trans{l} y'}{x + y \trans{l} y'} \textnormal{~, ~where } l \in L.
\]

\medskip

\begin{tabular}{ll}
\begin{minipage}[b]{0.5\linewidth}
\begin{alltt}\codefont
(fmod SPECIFICATION is including RULES .
  
  op 0   : -> PClosedTerm .
  op _._ : PLabel PTerm -> PTerm
           [metadata "sos"] .
  op _+_ : PTerm PTerm -> PTerm
           [metadata "sos"] .
  
  ops a b c : -> PAction .
  
  var x y x' y' : PTerm .
  var l : PLabel .
\end{alltt}
\end{minipage}
\begin{minipage}[b]{0.5\linewidth}
\begin{alltt}\codefont
  mb  
    ===
    l . x -(l)-> x

    x -(l)-> x'
    ===
    x + y -(l)-> x'
    
    y -(l)-> y'
    ===
    x + y -(l)-> y'  : PAllRules .
endfm)
\end{alltt}
\end{minipage}
\end{tabular}

\medskip

As illustrated, Maude provides good support for working with operators in infix notation.
As an improvement over {\PREG} \cite{AcetoCGI11a}, we use the generic label \code{l} of sort \code{PLabel} as syntactic sugar, instead of writing rules for each of the three concrete actions.
These actions are declared as constants, operations without a domain, of sort \code{PAction}, which is a subsort of \code{PLabel} described later starting with Section~\ref{sec:simbis}. (We can have other types of labels, not just actions.) The deadlock process is also declared as a constant of sort \code{PClosedTerm}, which stands for ${T}(\Sigma)$ and is declared internally as a subsort of \code{PTerm}.

\section{Meta SOS Components\label{sec:components}}

Though {\MetaSOS} is conceived as a general SOS framework, we have so far limited our development to case studies involving only GSOS systems \cite{Bloom95}.
These systems have certain desirable properties and, in spite of their restricted format, they cover most of the operations in the literature \cite{Aceto01}.

\begin{defn}[GSOS rule format]
\label{def:gsos}
Consider a process signature $\Sigma$. A \emph{GSOS rule} $\rho$ over $\Sigma$ has the shape:
\[
\sosrule{\{x_i \trans{l_{ij}} y_{ij}  \mid  i \in I, j \in I_{i}\} \cup \{x_{i} \ntrans{l_{ij}} \mid i \in J, j \in J_{i}\}}{f(\vec{x}) \xrightarrow{l} C[\vec{x},\vec{y}]}
,\]
where all variables are distinct, $f$ is an operation symbol from $\Sigma$ with arity $n$, $I,J \subseteq \{1, \ldots, n\}$, $I_{i}, J_{i}$ are finite index sets, the $l_{ij}$'s and $l$ are labels standing for actions ranging over the set $L$, and $C[\vec{x}, \vec{y}]$ is a $\Sigma$-context with variables including at most the $x_{i}$'s and $y_{ij}$'s.

A \emph{GSOS system} is a TSS $(\Sigma, L, {\cal R})$ such that $\Sigma$ and $L$ are finite, and ${\cal R}$ is a finite set of  rules in the GSOS format.
\end{defn}

The operational semantics of a TSS in the GSOS format is given in terms of a labelled transition system (LTS), whose transition relations are defined by structural induction over closed terms using the rules.
An essential property we make use of is that LTS's induced by GSOS systems are finitely branching \cite{Bloom95}. It is easy to see that BCCSP respects the GSOS format.

In what follows we present three main features provided by the {\MetaSOS} framework.

\subsection{Simulator and Bisimilarity Checker\label{sec:simbis}}

The purpose of the simulator associated to a TSS is to find all transitions for a given closed term.
Formally, the simulator finds, for a given closed term $p$, all the labels $l$ and closed terms $p'$ such that $p \trans{l} p'$. This is a slightly more general approach than the one in \cite{sos-maude}, where the user needs to give not only the initial term, but also the label as input.

To illustrate how to use the simulator, consider $\_\!\parallel\!\!\_$, the interleaving parallel composition without communication.
We want to also define its behaviour in the context of the termination predicate $\downarrow$.
As {\MetaSOS} does not currently provide direct support for working with predicates, unlike {\PREG}, we model predicate satisfiability by means of transitions. By adding the termination predicate trigger as a label, the rules for the prefix and choice operators remain the same. The rules for interleaving parallel composition are:

\[
\sosrule{x \trans{\alpha} x'}{x \parallel y \trans{\alpha} x' \parallel y}
\hspace{3ex}
\sosrule{y \trans{\alpha} y'}{x \parallel y \trans{\alpha} x \parallel y'}
\hspace{3ex}
\sosrule{x \trans{\downarrow} x ~~ y \trans{\downarrow} y'}{x \parallel y \trans{\downarrow} \nul}.
\]

Here $\alpha$ stands for any of the considered actions from the set $\{a,b,c\}$, but not for the termination predicate trigger. Also, we want to make sure that the last rule is applied only for $\downarrow$, but for none of the actions.
To specify this we enhance the previous specification with a new sort for predicates as a subset of labels, add the termination predicate, a variable ranging only over actions, and the rules:

\medskip

\noindent
\begin{tabular}{lll}
\begin{minipage}[b]{0.33\linewidth}
\begin{alltt}\codefont

sort PPredicate .
subsort PPredicate < PLabel .

x -(alpha)-> x'
===
x || y -(alpha)-> x' || y
\end{alltt}
\end{minipage}
\begin{minipage}[b]{0.33\linewidth}
\begin{alltt}\codefont
 op | : -> PPredicate .
 var alpha : PAction .
  
 y -(alpha)-> y'
 ===
 x || y -(alpha)-> x || y'
\end{alltt}
\end{minipage}
\begin{minipage}[b]{0.33\linewidth}
\begin{alltt}\codefont


x -(|)-> x' , y -(|)-> y'
===
x || y -(|)-> 0
\end{alltt}
\end{minipage}
\end{tabular}

\medskip

We first need to set up a simulator \code{(derive simulator SPECIFICATION .)}.
Not only does this prepare the metalanguage application to perform simulations, but also outputs a pure Maude specification that can be used outside the {\MetaSOS} environment for simulations within the specified system.
The advantage of using this generated simulator is a minor gain in performance due to the elimination of the overhead that comes with any metalanguage application.
In addition, this allows for the use of Maude tools such as the reachability analyzer and the LTL model checker. 

To perform a one step simulation for a given term we use the command \code{(simulate ...~.)}. For instance, the concrete call to observe how $p = \,\downarrow . \nul \parallel a . \nul$ is simulated is \code{(simulate | .~0 || a .~0~.)}. The output is a list of pairs of the shape $< l ~\#~ p' >$, where $l$ is a label of a provable transition and $p'$ is the resulting term. In our case the output is: \code{Possible steps:~< a \# | .~0 || 0 >}.
Note that due to our making a clear distinction between actions and predicates only one of the rules involving actions is applicable. The term $(\downarrow . \nul + b . \nul) \parallel (c . \nul + \downarrow . \nul)$, on the other hand, does involve all the specified rules:

\begin{alltt}\codefont
> (simulate | . 0 + b . 0 || c . 0 + | . 0 .)

Possible steps:
  < b \# 0 || c . 0 + | . 0 >
  < c \# | . 0 + b . 0 || 0 >
  < | \# 0 >
\end{alltt}

From the implementation perspective we tackled one of the issues suggested as future work in \cite{sos-maude}.
The caveat of the tool presented in that paper is that the user needs to provide term matching and substitution definitions by hand for every operator.
Our approach uses and extends Maude's meta-level functionality of working with substitutions in such a way that it becomes transparent to the user.

As the idea of rewrite-based SOS simulators has already been explored in \cite{MaudeMSOS,sos-maude,Verdejo2006226}, we focused only on performing one step simulations.
Having that functionality, it was natural to derive a strong bisimilarity checker that implements the following definition.

\begin{defn}[Strong Bisimilarity\ \cite{Park81}]\label{def:strongbisim}
Consider a TSS ${\cal T} = (\Sigma, L, {\cal R})$. A relation $R \subseteq T(\Sigma) \times T(\Sigma)$ is a \emph{strong bisimulation} if and only if it is symmetric and 
\[\forall_{p, q} (p,q) \in R \Rightarrow (\forall_{l,p'} p \trans{l} p' \Rightarrow \exists_{q'} q \trans{l} q' \land (q, q') \in R).\]
Two closed terms $p$ and $q$  are \emph{strongly bisimilar}, denoted by $p \bisimpar{{\cal T}} q$, if there exists a strong bisimulation relation $R$ such that $(p,q) \in R$. Whenever ${\cal T}$ is known from the context, we simply write $p \bisim q$.
\end{defn}

In \cite{Bloom95} it is shown that bisimilarity is a congruence for GSOS systems and that the LTS's defined using these systems are finitely branching.
These properties are necessary when checking for strong bisimilarity by means of the axiom schema that we will present in Section~\ref{sec:axiomderiver}.
Before that, let us present how to check strong bisimilarity using {\MetaSOS}.

In order to check if, for instance, $a . \nul \parallel b.\nul ~\bisim~ a . b . \nul + b . a . \nul$ holds, we use the command
\begin{alltt} \codefont
> (check (a . 0 || b . 0) \(\sim\) (a . b . 0 + b . a . 0) .)

result: true.Bool
\end{alltt}


The Maude specification output when setting up the simulator also includes the bisimilarity checker.
Running this specification allows one to directly use the functions that implement the simulator and bisimilarity checker features, which have the same name as the user interface commands.
These are called using \code{reduce} in the core Maude environment: \code{reduce simulate  ...~.} and \code{reduce check ...~$\sim$ ...~.}

The bisimilarity checker does not currently handle process terms with infinite behaviour. The module presented in the next section, however, can check if two terms from $\Sigma_{\BCCSP}$, defined using guarded recursion, are bisimilar.

\subsection{Axiom Schema Deriver \label{sec:axiomderiver}}

As an alternative method for reasoning about strong bisimilarity, {\MetaSOS} includes a component for generating axiom schemas that are sound and ground-complete modulo bisimilarity. 
There has been a notable amount of effort put into developing algorithms for axiomatizations for GSOS-like systems \cite{Aceto94, AcetoCGI2011, BaetenV04}, yet all involve several transformations of the original system before deriving the axioms.
After implementing one such algorithm in the tool {\PREG} \cite{AcetoCGI11a}, a simpler method was developed in \cite{gsos-data}. We slightly adapt that approach here by using an extended version of the prefix operation and by also showing how to axiomatize operations defined using rules with negative premises, not just positive ones.

When given a signature $\Sigma$ that includes $\Sigma_{\BCCSP}$, the purpose of an axiomatization of strong bisimilarity is to rewrite each term $t \in T(\Sigma)$, that is semantically well founded in the sense of Definition~5.1 from \cite{Aceto94}, to another term $t'$ such that $t \bisimpar{} t'$ and $t' \in T(\Sigma_{\BCCSP})$.
This reduces the problem of axiomatizing bisimilarity over $T(\Sigma)$ to that of axiomatizing it over ${\BCCSP}$.
It is well known \cite{DBLP:journals/jacm/HennessyM85} that the following axiomatization (denoted by $E_{\BCCSP}$) is sound and ground-complete for bisimilarity on ${\BCCSP}$:

\begin{center}
\begin{tabular}{r@{\hspace{3pt}}c@{\hspace{3pt}}lr@{\hspace{20pt}}r@{\hspace{3pt}}c@{\hspace{3pt}}lr}
$x + y$ & = & $y + x$ & ~ &
$x + x$ & = & $x$ & ~
\\[1ex]
$(x + y) + z$ & = & $x + (y + z)$ & ~ &
$x + \nul$ & = & $x$\, &
\end{tabular}
\end{center}

In order to set up Maude to perform equational reasoning using $E_{\BCCSP}$, we can declare that $\_+\!\_$ is associative and commutative so that rewrites are performed modulo these two properties: 
\code{op }{{\codefont\verb#_+_#}}\code{ :~PTerm PTerm -> PTerm [assoc comm metadata "sos"] .} Also, even though we could specify idempotence and identity element as attributes, for performance reasons we add the last two equations explicitly to the specification: \code{eq x + x = x .~eq x + 0 = x .}
For convenience {\MetaSOS} already includes a module with the signature and equations for BCCSP named \code{ET-BCCSP} that can be included in the specification.
For this reason, the names \code{.} and \code{+} are reserved, which means that if the user wants to specify his/her own version of the prefix and choice operations some other names need to be used.

\begin{defn}[{Head} Normal Form]
\label{def:hnf}
Let $\Sigma$ be a signature such that $\Sigma_{\BCCSP} \subseteq \Sigma$.
A term $t$ in $\mathbb{T}(\Sigma)$ is in \emph{head normal form} (for short, h.n.f.) if
\(
t = {\sum_{i \in I}l_i.t_{i}}.
\) The empty sum $(I = \emptyset)$ is denoted by the deadlock constant $\nul$.
\end{defn}

\begin{defn}[Disjoint extension]
\label{def:disjExt}
A {GSOS} system $G'$ is a disjoint extension of a {GSOS} system $G$, written $G \sqsubseteq G'$, if the signature and the rules of $G'$ include those of $G$, and $G'$ does not introduce new rules for operations in $G$.
\end{defn}

\begin{defn}[Axiomatization schema]\label{def:axiomatization} Let ${\cal T} = (\Sigma, L, {\cal R})$ be a TSS in GSOS format  such that $\textnormal{\BCCSP} \sqsubseteq {\cal T}$. By $E_{{\cal T}}$ we denote the axiom system that extends $E_{\BCCSP}$ with the following axiom schema for every operation $f$ in ${\cal T}$, parameterized over the vector of closed process terms $\vec{p}$ in h.n.f.:

\begin{center}
\noindent
$f(\vec{p}) = \sum \left\{l.C[\vec{p},\vec{q}] \,\,\Bigg|\,\, \rho = \sosrule{H}{f(\vec{x}) \xrightarrow{l} C[\vec{x},\vec{y}]} \in {\cal R}, \vec{p} = \sigma(\vec{x}), \vec{q} = \sigma(\vec{y}) \textnormal{ and } \checkmark(\vec{p}, \rho) \right\}$,
\end{center}

\noindent
\hspace{10ex} where $\checkmark$ is defined as
~~\(
\checkmark(\vec{p}, \rho) = \bigwedge_{p_{k} \in \vec{p}} ~\checkmark'(p_{k}, k, \rho),\)

\noindent
\hspace{10ex}
\(
\begin{array}{@{}l@{}l}
\text{and }~~
\checkmark' &\left( p_{k}, k, \sosrule{\{x_i \trans{l_{ij}} y_{ij}  \mid  i \in I, j \in I_{i}\} ~ \{x_{i} \ntrans{l_{ij}} \mid i \in J, j \in J_{i}\}}{f(\vec{x}) \xrightarrow{l} C[\vec{x},\vec{y}]} \right) =\\[3ex]
&~\textnormal{if } k \in I \textnormal{ then } \forall_{j \in J_{k}}~ \exists_{p', p'' \in T(\Sigma_{P})} ~ p_{k} \equiv l_{k j}.p' + p'' \textnormal{ and}\\
&~\textnormal{if } k \in J \textnormal{ then } \forall_{j \in J_{k}}~ \forall_{p', p'' \in T(\Sigma_{P})} ~ p_{k} \not\equiv l_{kj}.p' + p'',\\
&~\textnormal{where } \equiv \textnormal{ denotes equality up to } E_{\BCCSP}.
\end{array}
\)

\end{defn}

Intuitively, the axiom transforms $f(\vec{p})$ into a sum of closed terms covering all its execution possibilities. This is akin to Milner's well known expansion law for parallel composition of head normal forms. In order to obtain them we iterate through the set of $f$-defining rules and check if $\vec{p}$ satisfies their hypotheses by means of $\checkmark$. The predicate $\checkmark$ makes sure that, for a given rule, every component of $\vec{p}$ is a term with enough action prefixed summands satisfying all the positive premises that involve the component, and no summands prefixed with the actions from any of the corresponding negative premises.


\begin{thm}
\label{thm:axiomsoundcomplete}
Consider a TSS ${\cal T} = (\Sigma, L, {\cal R})$ that is semantically well founded in the sense of  \cite[Definition~5.1]{Aceto94}, such that ${\BCCSP} \sqsubseteq {\cal T}$. $E_{{\cal T}}$ is sound and ground-complete for strong bisimilarity on $T(\Sigma_{P})$.
\end{thm}
\begin{proof} Soundness follows in the standard fashion. Every transition $f(\vec{p})$ can perform is matched by the right hand side of the equation and vice versa due to the natural derivation of the execution tree according to the defining rules.

As shown in \cite{Aceto94}, in order to prove ground-completeness of an axiom system, it is sufficient to show that it is head normalizing, which means that it can bring any closed term to a h.n.f. Note that the axiomatization presented in Definition~\ref{def:axiomatization} always derives terms in h.n.f.
\end{proof}

In order to generate a Maude equational theory for the operations in a specification we use the command \code{(derive axiom schemas SPECIFICATION .)}.
Just like in the case of the simulator component, the command both prepares the environment to perform equational reductions according to the generated axioms, and outputs a Maude specification that can be used externally, independently of the environment.
The generated equational theory has the name of the specification with the suffix ``\code{-SCHEMA}'' and is selected using the command \code{(select SPECIFICATION-SCHEMA .)}.

The standard command \code{reduce} derives the normal form of a given closed term. For example, having loaded the specification of $\_\!\parallel\!\!\_$ from Section~\ref{sec:simbis}, this is how we obtain the normal form of $a . \nul \parallel b.\nul$:
\begin{alltt}\codefont
> (reduce a . 0 || b . 0 .)

result PClosedTerm :  a . b . 0 + b . a . 0
\end{alltt}

To illustrate what the axiomatizations look like, consider a general binary operation between labels \textit{mix} and an operation $g$ defined as:

\[
\sosrule{x \trans{k} x' ~~ y \trans{l} y' ~~ x \ntrans{l} ~~ y \ntrans{k}}{g(x,y) \trans{\textit{mix(k,l)}} x' + y'} ~~~~
\sosrule{x \trans{l} x' ~~ y \trans{l} y'}{g(x,y) \trans{\textit{l}} \nul}.
\]

The first equation derives a sum of new operations, $g_{1}$ and $g_{2}$, one for each rule defining $g$.
These new operations have the same domain as $g$, only extended with one parameter that will ultimately hold the tree of all execution paths that start with the corresponding rule -- its head normal form.
Initially this parameter is set to $\nul$.

\begin{alltt}\codefont
  eq g(x,y) = g1(x,y,0) + g2(x,y,0) .
\end{alltt}

Let us first present the axiom for $g_{1}$, which is given as a standard Maude conditional equation, and then discuss all of its aspects.

\begin{alltt}\codefont
 ceq g1(x,y,SOLUTION) = g1(x,y,SOLUTION')  --- (0)
  if k . x' + x1 := x + dummy              --- (1)
  /\tb l . y' + x2 := y + dummy              --- (2)
  /\tb not(x can l)                          --- (3)
  /\tb not(y can k)                          --- (4)
  /\tb NEW-SUMMAND := mix(k, l).(x' + y')    --- (5)
  /\tb SOLUTION' := SOLUTION + NEW-SUMMAND   --- (6)
  /\tb SOLUTION =/= SOLUTION'                --- (7)
  .
\end{alltt}

Condition (1) requires that the first parameter satisfies the formula $x \trans{k} x'$. The variable $x$ needs to be matched by a term that has $k.x'$ as a summand ($x_{1}$ is a generic variable of sort \code{PTerm}). If $x$ is exactly of the shape $k.\_$ then Maude cannot find a match between $k . x' + x_{1}$ and $x$, not knowing that it can assign $\nul$ to $x_{1}$, which explains the use of a constant of sort \code{PTerm} denoted by \code{dummy} added as a summand on the right hand side.

Condition (3) requires that $x$ satisfies the formula $x \ntrans{l}$. The natural inductive definition of the operation \textit{can} is included in the module \code{ET-BCCSP}:

\medskip
\begin{alltt}\codefont
  op _can_ :  PTerm PLabel -> Bool .

  eq 0 can l = false .                         eq (l . x) can l = true .
  eq (x + y) can l = (x can l) or (y can l) .  ceq (l . x) can k = false if l =/= k .
\end{alltt}
\medskip

If conditions (1)\,--\,(4) are satisfied, then the premises of $g$'s first rule are met. This means that $\textit{mix}(k, l).(x' + y')$, set at line (5) as the value for the variable \code{NEW-SUMMAND}, has to be a summand of the resulting head normal form. This head normal form is computed incrementally, by finding such summands individually, using the third parameter of $g_{1}$: \code{SOLUTION} and \code{SOLUTION'} hold the head normal forms computed before and, respectively after the current call of $g_{1}$.
The aforementioned summand is added only if it is not already part of \code{SOLUTION} (conditions (6)\,--\,(7)). Should all conditions hold, a recursive call of $g_{1}$ is initiated (line (0)).

An important fact to keep in mind is that, in the specification, for any given rule, the labels of negative transitions given as variables need to also appear on some of the positive ones. For instance, 
had we not had the premise $x \trans{k} x'$, where $k$ is a variable over the set of labels $L = \{a,b,c\}$, it would have been impossible to tell if condition (4) is met due to the missing assignment for $k$ that should have resulted when evaluating condition (1).
It is possible, however, to have a rule with negative premises labelled directly with constants, without the need for those constants to appear in other premises of the same rule.

If any of the conditions (1)\,--\,(4) does not hold, or if no new solution is found, then the following base-case equation is called:

\begin{alltt}\codefont
  eq g1(x,y,SOLUTION) = SOLUTION [owise] .
\end{alltt}

The equations for $g_{2}$ are generated in a similar fashion:
\begin{alltt}\codefont
  ceq g2(x,y,SOLUTION) = g2(x,y,SOLUTION')
    if l . x' + x1 := x + dummy     /\tb NEW-SOLUTION := l . 0
    /\tb l . y' + x2 := y + dummy     /\tb SOLUTION' := SOLUTION + NEW-SOLUTION
                                    /\tb SOLUTION =/= SOLUTION' .
   eq g2(x,y,SOLUTION) = SOLUTION [owise] .
\end{alltt}

In {\MetaSOS} one can also specify recursive processes.
If two such processes are given in {\BCCSP} with guarded recursion in order to determine whether they are bisimilar, the user can call a decision procedure implementing a unique fixed point induction algorithm. Currently the user needs to make sure that the guardedness condition is met.

By way of example, consider the following transition systems.
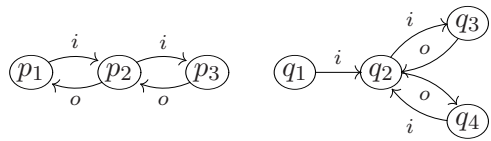
\[\xymatrix@R=.2cm@C=.6cm{
& & & & &  *+[o][F]{{q_3}}\ar@/^/[dl]_{o}\\
*+[o][F]{p_1}\ar@/^/[r]^{i} &
*+[o][F]{p_2}\ar@/^/[l]^{o}\ar@/^/[r]^{i} &
*+[o][F]{p_3}\ar@/^/[l]^{o} &
*+[o][F]{q_1}\ar[r]^{i} &
*+[o][F]{q_2}\ar@/^/[ur]^{i}\ar@/^/[dr]_{o} \\
& & & & &  *+[o][F]{{q_4}}\ar@/^/[ul]^{i}
}\]

We specify this behaviour in {\MetaSOS} by means of the reserved operation \code{def}:

\begin{alltt} \codefont
 ops p1 p2 p3 q1 q2 : -> PClosedTerm .  ops i o : -> PAction .
  
 eq def(p1) = i . p2 .           eq def(q1) = i . q2 .           eq def(q3) = o . q2 .
 eq def(p2) = i . p3 + o . p1 .  eq def(q2) = i . q3 + o . q4 .  eq def(q4) = i . q2 .
 eq def(p3) = o . p2 .
\end{alltt}

The command \code{(reduce areEqual(p1, q1) .)} checks whether $p_{1}$ and $q_{1}$ are bisimilar. The output in this case is the pair \code{< true ; < p1 ; q1 > < p1 ; q4 > < p2 ; q2 > < p3 ; q3 > >}, where the first element of the pair indicates whether the processes are bisimilar, and, if this is indeed the case, the second one is a representation of the found bisimulation.

\subsection{Commutativity Format Checker\label{sec:formatchecker}}

Besides automatically deriving sound and ground-complete axiomatizations, the focus of {\MetaSOS} is also to check for algebraic properties of operations, by design.
We have implemented a component that analyzes the provided SOS specification in order to find binary operations that are commutative.
We adapt the format for binary operations from \cite{Mousavi05-IPL} to GSOS systems that may have negative premises.

\begin{defn}[Commutativity]
Given a TSS and a binary process operator $f$ in its process signature, $f$ is called
{\em commutative w.r.t. a relation $\sim$}, if the following equation is sound w.r.t.\
$\sim$:
\[
f(x_0, x_1) = f(x_1, x_0).
\]
\end{defn}

\begin{defn}[\label{def::commform}Commutativity format \cite{axiom-algebraic,Mousavi05-IPL}]
A transition system specification over signature $\Sigma$ is in
\commform\ format with respect to a set of binary function symbols
$\COMM \subseteq \Sigma$ if all its $f$-defining transition rules with
$f \in \COMM$ have the following form
\[
\DedRule{c}~\sosrule{\{x_i \trans{l_{ij}} y_{ij}  \mid  i \in \{0,1\}, j \in I_{i}\} \cup \{x_{i} \ntrans{l_{ij}} \mid i \in \{0,1\}, j \in J_{i}\}}{f(x_{0}, x_1) \trans{l} t}
\]
where $I_{i}$ and $J_{i}$ are finite index sets for each $i \in \{0,1\}$, and variables
appearing in the source of the conclusion and target of the premises
are all pairwise distinct. We denote the set of premises of
\DedRule{c} by $H$. Moreover, for each such rule,
there exist a transition rule \DedRule{c'} of the following form in
the transition system specification
\[
\DedRule{c'}~\sosrule{H'}{f(x'_{0}, x'_{1}) \trans{l} t'}
\]
and a bijective mapping (substitution) $\hbar$ on variables such that (1)
 $\hbar(x'_0) = x_1$ and  $\hbar(x'_1) = x_0$,
(2)
 $\hbar(t') \simcom t$ and
(3)
 $\hbar(h') \in H$, for each ${h' \in H'}$.
Here $\simcom$ means equality up to swapping of arguments of operators in $\COMM$ in any context.
Transition rule \DedRule{c'} is called the {\em commutative mirror} of \DedRule{c}.
\end{defn}

\begin{thm}[Commutativity for \textsf{comm-form} \cite{axiom-algebraic,Mousavi05-IPL}]
\label{thm:commform}
If a transition system specification is in \commform\ format with respect to a set of operators $\COMM$, then all operators in $\COMM$ are commutative with respect to strong bisimilarity.
\end{thm}

We implement an algorithm that, for a given operation, searches for all of its rules that are commutative mirrors.
It is well know that parallel composition is commutative.
To check this using our tool we load the specification presented in Section~\ref{sec:simbis} and call \code{(check formats SPECIFICATION .)}.
The output shows that the first two rules defining $\_\!\parallel\!\!\_$ are commutative mirrors, and that the third rule involving the termination predicate $\downarrow$ is a commutative mirror of itself, by pointing out the bijective mapping:

\begin{alltt}\codefont
  _||_ is commutative:

  x -(alpha)-> x'                          y -(alpha)-> y'
  ===                          mirrors     ===
  x || y -(alpha)-> x' || y                x || y -(alpha)-> x || y'

  with:  alpha <- alpha  x' <- y'  x <- y  y' <- x'  y <- x

  x -(|)-> x' , y -(|)-> y'                x -(|)-> x' , y -(|)-> y'
  ===                          mirrors     ===
  x || y -(|)-> 0                          x || y -(|)-> 0

  with:  x' <- y'  x <- y  y' <- x'  y <- x
\end{alltt}

What {\MetaSOS} does internally is to generate a Maude theory that has the name of the specification with the suffix ``\code{-FORMATS}''.
It is the same as the initial specification, only that all the \code{"sos"} operators that are found commutative are enhanced with the attribute \code{comm}.
This is of use both when having to perform rewrites modulo commutativity involving those operations, and as meta-information for future components that may need it.
One of these components could, for instance, be dedicated to optimizing axiomatizations, using the approach presented in \cite{axiom-algebraic}.

An important thing to remark is that the label mapping \code{alpha <- alpha} appears amongst the process variables mapping. The reason we extend the mapping to labels too is the fact that the user should not be forced to use the same variable name for matching premises of different rules.
We would thus find that the first two rules are commutative mirrors even if they had different variables for actions, e.g. \code{alpha} and \code{beta}, respectively.

Aside from giving the user more freedom when choosing names for label variables, extending the mapping to labels is actually necessary for proving that some operators are commutative.
Consider, for example, the operation $g$ introduced in Section~\ref{sec:axiomderiver} and assume that the operation \textit{mix} over labels is declared as commutative.
Suppose label variables were not taken into account when searching for commutative mirrors. Then there would be no way of directly proving that $g$ is commutative, unless the user specified the 6 instantiations of the first rule for $g$, involving the concrete action labels $a,b,c$.

\begin{alltt}\codefont
> (check formats SPECIFICATION .)

g is commutative:

x-(k)->x',y-(l)->y',x-(l)/>,y-(k)/>             x-(k)->x',y-(l)->y',x-(l)/>,y-(k)/>
===                                   mirrors   ===
g(x,y) -(mix(k,l))-> x' + y'                    g(x,y) -(mix(k,l))-> x' + y'

with: k <- l  l <- k  x' <- y'  x <- y  y' <- x'  y <- x


x -(l)-> x' , y -(l)-> y'           x -(l)-> x' , y -(l)-> y'
===                        mirrors  ===
g(x,y) -(l)-> 0                     g(x,y) -(l)-> 0

with: l <- l  x' <- y'  x <- y  y' <- x'  y <- x
\end{alltt}

If we look at the first rule, note that when applying the substitution on labels, in order to check for the commutativity format, we need to make sure that \code{mix(k,l)} and \code{mix(l,k)} stand for the same label. This holds in our case because we do not merely check for syntactic equality, but for equality within the algebra defined for labels. Recall that we consider \textit{mix} to be commutative.

The first rule is found as a mirror of itself based on the commutativity of {\codefont \verb#_+_#}.
Had the consequent of the rule been of the shape \code{x' * y'} ({\codefont \verb#_*_#} being a new binary operation), {\MetaSOS} would have attempted to prove first that {\codefont \verb#_*_#} is commutative.

\subsection{Linda -- Integrating Components\label{sec:linda}}
\newcommand{\sdspace}{\textnormal{$t_{D}$}}
\newcommand{\nr}[1]{(#1) \hspace*{0.1cm}}
\newcommand{\Seq}{\mathrel{;}}

\newcommand{\ask}{\textit{ask}}
\newcommand{\nask}{\textit{nask}}
\newcommand{\tell}{\textit{tell}}
\newcommand{\get}{\textit{get}}

In this section we present another case study and show how easy it is to make use of the functionality provided by all the previously described components.
Let us focus on the tuple-space based coordination language Linda \cite{Brogi98} and its SOS semantics, as given in \cite{Mousavi05-IC}.

Consider a minimalistic signature for the data component, $\Sigma_{D}$, that consists of constants for tuples (typical members $u,v$) and two operations for working with multisets of tuples: $\emptyset$ for the empty multiset and $\_\_$ (blank) as a commutative and associative binary separator for the elements from the multiset. The operation
$\_\_$ has $\emptyset$ as identity element.
We prefer to use constructs instead of the standard mathematical ones (braces ``\{'', ``\}'' for set separators, commas ``,'' for separating elements within a set, and set union operator ``$\cup$'') for implementation purposes.
For instance, the multiset $\{u,v\} \cup \{u\} \cup \emptyset$ is written as  $u ~ v ~ u ~ \emptyset$ in ${\mathbb{T}}(\Sigma_{D})$, which is the same as $u ~ v ~ u$ because $\emptyset$ is the identity element.
(That is actually the standard Maude notation for sets and multisets.)
This is how we declare the above mentioned signature:

\begin{alltt}\codefont
  sort PData PClosedData .  subsort PClosedData < PData  .
  op  empty : -> PData .
  op     __ : PData PData -> PData [assoc comm id: empty] .
  ops   u v : -> PClosedData .
\end{alltt}

Linda has several constructs for manipulating a shared data component of the language:
\begin{itemize}
\item $\ask(u)$ and $\nask(u)$ check, respectively, whether tuple $u$ is (or is not) in the data space,
\item $\tell(u)$ adds tuple $u$ to the data space,
\item $\get(u)$ removes tuple $u$ from the data space.
\end{itemize}

The $\ask(u)$ and $\get(u)$ operations are blocking, in the sense that a process executing them blocks if $u$ is not in the data space.

In \cite{gsos-data} we show how to use labels for operating with the data component.
For Linda, the set of labels $L$ is extended to triples of the from $\langle\sdspace, -, \sdspace'\rangle$, where $\sdspace, \sdspace'$ are open data terms from ${\mathbb{T}}(\Sigma_{D})$, standing for the store before and, respectively, after the transition. The language does not have actions, hence the use of the placeholder ``$-$'' within the triple.
As shown later, in order to have a finite set of labels and rules, which is necessary to have a proper GSOS system, we use symbolic names instead of open data terms.

Besides the four constructs for operating with the store, the language includes the prefix operation $l.\_$ (for every $l$ in $L$), nondeterministic choice $\_+\!\_$, parallel composition $\_\!\parallel\!\!\_$, and sequential composition $\_\!\Seq\!\!\_$, all in the context of the already introduced termination predicate $\downarrow$. Linda also comes with a successfully terminated process, which we denote by $\downarrow . \nul$.

In order to handle the store, our approach of extending the prefix operation to triples is slightly different from the one in \cite{gsos-data}. Though less intuitive, it is easier to implement than the one involving two new operations, \textit{check} and \textit{update}, because it requires no extra core axioms aside from those in $E_{\BCCSP}$.

\begin{alltt}\codefont
  op <_,_,_> : PData PAction PData -> PLabel .  op - : -> PAction .
\end{alltt}

We first make sure that the SOS specification disjointly extends {\BCCSP}, as required by Theorem~\ref{thm:axiomsoundcomplete}. The rules for ${\Sigma_{\BCCSP}}$ are declared as presented in Section~\ref{sec:pre} because they are the same both for the extended labels and the termination predicate.

Given that $\mu$ is a variable to be replaced by any considered constant tuple, the rules for the operations manipulating the data component are:
\[
\begin{array}{c}
 \sosrule{}{\ask(\mu) \trans{\langle{x_{D} ~ \mu},-,{x_{D} ~ \mu}\rangle} {\downarrow.\nul}}
\qquad
 \sosrule{}{{\tell(\mu)} \trans{\langle{x_{D}},-,{x_{D} ~ \mu}\rangle} {\downarrow.\nul}}
\qquad
 \sosrule{}{{\get(\mu)} \trans{\langle{x_{D}  ~ \mu},-,{x_{D}}\rangle} {\downarrow.\nul}}.
\end{array}
\]

Linda also has a basic operation named {\nask} that checks if a tuple is not in the tuple space. The operation, however, is defined using side conditions, and currently we provide no support for such rules.

For the purpose of demonstration, we will only implement a limited and artificial version of Linda:

\begin{alltt}\codefont
ops ask tell get : PClosedData -> PTerm
                   [metadata "sos"] .      
op  d : -> PData .                         ===
var mu : PClosedData .                     ask(mu) -( <(d mu), -, (d mu)> )-> |.0
                                           
===                                        ===
tell(mu) -(<d, -, (d mu)>)-> |.0           get(mu) -( <(d mu), -, d> )-> |.0
\end{alltt}

The limitation consists in the use of a symbolic constant \code{d}, denoting a data term, instead of a variable of the same sort. This is because in \cite{gsos-data} it is presented how to derive a sound and ground-complete axiomatization modulo a notion of bisimilarity only for systems with a data component whose domain is a finite set of constants, and not a (possibly infinite) set of open terms.
In our case the domain of the data component can be thought of as a set of constants, limited to the number of tuples taken into account plus one (for the symbolic constant).
Using the constant \code{d} is also useful during the axiomatization process because it helps avoiding generating equations with fresh variables on the right hand side.
For instance, according to the schema from Definition~\ref{def:axiomatization}, the following axiom $\tell(\mu) = \langle d,-,d~\mu\rangle . \downarrow . \nul$ is generated, and here it is required that $d$ is not a variable.

The rules for $\_\!\parallel\!\!\_$ are very similar to those shown in Section~\ref{sec:simbis}, and those for $\_\!\Seq\!\!\_$ are:

\[
\begin{array}{c}
 \sosrule{{x} \trans{\langle{x_{D}},-,{x_{D}'}\rangle} {x'}}{{x \Seq y} \trans{\langle{x_{D}},-,{x_{D}'}\rangle} {x' \Seq y}}
\qquad
 \sosrule{{x} \trans{\downarrow} {x'} ~~~ {y} \trans{\langle{x_{D}},-,{x_{D}'}\rangle} {y'}}{{x \Seq y} \trans{\langle{x_{D}},-,{x_{D}'}\rangle} {y'}}
\qquad
 \sosrule{{x} \trans{\downarrow} {x'} ~~~ {y} \trans{\downarrow} {y'}}{{x \Seq y} \trans{\downarrow} {\nul}}
\end{array}.
\]

The rules for the last two operations do not introduce new names for data terms on the consequent transitions (all the names are known from the premises), which means that no axioms with fresh variables on the right hand side can be generated. Therefore it is safe to declare them using variables of sort \code{PData} instead of symbolic constants.

\begin{alltt}\codefont
op _;_ : PTerm PTerm -> PTerm [metadata "sos"] .   var xD xD' : PData .

x -(<xD,-,xD'>)-> x'        x -(|)-> x', y -(<xD,-,xD'>)-> y'  x -(|)-> x', y -(|)-> y'
===                         ===                                ===
x;y -(<xD,-,xD'>)-> (x';y)  x;y -(<xD,-,xD'>)-> y'             x;y -(|)-> y'
\end{alltt}

A use case scenario involving all the components illustrated so far may start with loading the specification for Linda and checking which operations are commutative \code{(check formats LINDA .)}. Remark that {\codefont \verb#_;_#}'s commutativity cannot be proven:

\begin{alltt}\codefont
Could not prove commutativity for:  _;_
Could not find commutative mirrors within:
x -(<xD,-,xD'>)-> x'        x -(|)-> x', y -(<xD,-,xD'>)-> y'  x -(|)-> x', y -(|)-> y'
===                         ===                                ===
x;y -(<xD,-,xD'>)-> (x';y)  x;y -(<xD,-,xD'>)-> y'             x;y -(|)-> y'
\end{alltt}

We could continue by deriving the axiom schema and determining the normal form of a term such as $\ask(u) \Seq \tell(v)$. Finally we can check if indeed the found normal form is bisimilar to the initial term.

\begin{alltt}\codefont
> (derive axiom schemas LINDA-FORMATS .)
> (select LINDA-FORMATS-SCHEMA .)
> (reduce ask(u) ; tell(v) .)
result PClosedTerm :  < d u,-,d u > . < d,-,d v > . | . 0
> (derive simulator LINDA-FORMATS .)
> (check (ask(u) ; tell(v)) \(\sim\) (< d u,-,d u > . < d,-,d v > . | . 0) .)
result: true.Bool
\end{alltt}
\subsection{Adding Components\label{sec:addingcomponents}}

{\MetaSOS} is conceived in a way to be easily extended with new components.
Besides the three components presented in Sections \ref{sec:simbis}, \ref{sec:axiomderiver}
and \ref{sec:formatchecker} the tool includes a file named \code{component-sample.maude} which the user can adapt to implement a new desired functionality by following the patterns presented in \cite{GoriacCLAG09}.

In what follows, the name ``sample'' is generic and is meant to be replaced by some other name suggesting the functionality of a new component.
Each component has two modules \code{SAMPLE-LANG-SIGN} and \code{SAMPLE-STATE-HANDLING}, dedicated for the signature of the implemented commands and, respectively, their semantics. Once implemented, the functionality is included in the {\MetaSOS} framework by following these steps:
(1) in the file \code{metasos-interface.maude} the signature for the new commands needs to be included in the \code{METASOS-LANG-SIGN} module and their semantics needs to be included in the module \code{METASOS-STATE-HANDLING},
(2) in the file \code{metasos.maude} the new component needs to be loaded just like the others: \code{load component-sample.maude}.

It is worth mentioning that, in order to ease the development cycle, the framework provides support for unit testing. It is beyond the scope of this paper, though, to present how to make use of this facility.
\section{Conclusion and Future Work\label{sec:conc}}

{\MetaSOS} addresses many of the extensions foreseen in \cite{sos-maude}. Namely, it represents a core framework dedicated to implementing SOS meta-theorems, it provides support for generating axiomatizations, and it frees the user from implementing matching procedures for specified language constructs.
In its present form, {\MetaSOS} can handle languages whose operational specification is in the GSOS format, such as most classic process calculi and Linda.
Another aspect addressed in \cite{sos-maude} is the support for more general SOS frameworks that allow for terms as labels, as well as multi-sorted and binding signatures.
This would allow the framework to handle name-passing and higher-order languages such as the $\pi$-calculus~\cite{SangiorgiW2001}.
Though {\MetaSOS} does not provide this kind of support yet, the general way in which it handles labels is a good step towards that goal.

There are, naturally, many ways to improve and extend the tool.
Besides checking for the commutativity format, there are many other formats to check for: determinism and idempotence \cite{AcetoBIMR12, AcetoGI13}, zero and unit elements \cite{AcetoCIMR11a}, associativity \cite{Mousavi08-CONCUR}, and distributivity \cite{AcetoCIMR11}.
Adapting {\PREG} and adding it as a component to {\MetaSOS} as presented in Section~\ref{sec:addingcomponents} would also be of value due to its different approach to generating axiomatizations, and because it includes a GSOS format checker.
The axiomatization process could be enhanced using the technique presented in \cite{axiom-algebraic}. This would lead to smaller and more natural axiom systems.
It would also be of interest to investigate axiomatizations that are sound and ground-complete modulo other notions of equivalence, such as weak bisimilarity.
From the theoretical viewpoint, it is worth investigating if the results on axiomatizations can be extended to coalgebras and whether a framework for SOS using the bialgebraic approach can be developed.


\bibliographystyle{eptcs}

\begin{thebibliography}{10}
\providecommand{\bibitemdeclare}[2]{}
\providecommand{\surnamestart}{}
\providecommand{\surnameend}{}
\providecommand{\urlprefix}{Available at }
\providecommand{\url}[1]{\texttt{#1}}
\providecommand{\href}[2]{\texttt{#2}}
\providecommand{\urlalt}[2]{\href{#1}{#2}}
\providecommand{\doi}[1]{doi:\urlalt{http://dx.doi.org/#1}{#1}}
\providecommand{\bibinfo}[2]{#2}

\bibitemdeclare{article}{AcetoBIMR12}
\bibitem{AcetoBIMR12}
\bibinfo{author}{Luca \surnamestart Aceto\surnameend}, \bibinfo{author}{Arnar
  \surnamestart Birgisson\surnameend}, \bibinfo{author}{Anna \surnamestart
  Ing{\'o}lfsd{\'o}ttir\surnameend}, \bibinfo{author}{Mohammad~Reza
  \surnamestart Mousavi\surnameend} \& \bibinfo{author}{Michel~A. \surnamestart
  Reniers\surnameend} (\bibinfo{year}{2012}): \emph{\bibinfo{title}{Rule
  formats for determinism and idempotence}}.
\newblock {\sl \bibinfo{journal}{Science of Computer Programming}}
  \bibinfo{volume}{77}(\bibinfo{number}{7--8}), pp. \bibinfo{pages}{889--907},
  \doi{10.1016/j.scico.2010.04.002}.

\bibitemdeclare{article}{Aceto94}
\bibitem{Aceto94}
\bibinfo{author}{Luca \surnamestart Aceto\surnameend}, \bibinfo{author}{Bard
  \surnamestart Bloom\surnameend} \& \bibinfo{author}{Frits~W. \surnamestart
  Vaandrager\surnameend} (\bibinfo{year}{1994}): \emph{\bibinfo{title}{Turning
  {SOS} rules into equations}}.
\newblock {\sl \bibinfo{journal}{Information and Computation}}
  \bibinfo{volume}{111}, pp. \bibinfo{pages}{1--52},
  \doi{10.1006/inco.1994.1040}.

\bibitemdeclare{inproceedings}{AcetoCGI2011}
\bibitem{AcetoCGI2011}
\bibinfo{author}{Luca \surnamestart Aceto\surnameend},
  \bibinfo{author}{Georgiana \surnamestart Caltais\surnameend},
  \bibinfo{author}{Eugen-Ioan \surnamestart Goriac\surnameend} \&
  \bibinfo{author}{Anna \surnamestart Ing{\'o}lfsd{\'o}ttir\surnameend}
  (\bibinfo{year}{2011}): \emph{\bibinfo{title}{Axiomatizing {GSOS} with
  Predicates}}.
\newblock In \bibinfo{editor}{Michel~A. \surnamestart Reniers\surnameend} \&
  \bibinfo{editor}{Pawel \surnamestart Sobocinski\surnameend}, editors: {\sl
  \bibinfo{booktitle}{Proceedings Eighth Workshop on Structural Operational
  Semantics 2011}}, {\sl \bibinfo{series}{EPTCS}}~\bibinfo{volume}{62}, pp.
  \bibinfo{pages}{1--15}, \doi{10.4204/EPTCS.62.1}.

\bibitemdeclare{inproceedings}{AcetoCGI11a}
\bibitem{AcetoCGI11a}
\bibinfo{author}{Luca \surnamestart Aceto\surnameend},
  \bibinfo{author}{Georgiana \surnamestart Caltais\surnameend},
  \bibinfo{author}{Eugen-Ioan \surnamestart Goriac\surnameend} \&
  \bibinfo{author}{Anna \surnamestart Ing{\'o}lfsd{\'o}ttir\surnameend}
  (\bibinfo{year}{2011}): \emph{\bibinfo{title}{{PREG Axiomatizer} - A Ground
  Bisimilarity Checker for {GSOS} with Predicates}}.
\newblock In \bibinfo{editor}{Andrea \surnamestart Corradini\surnameend},
  \bibinfo{editor}{Bartek \surnamestart Klin\surnameend} \&
  \bibinfo{editor}{Corina \surnamestart C\^{\i}rstea\surnameend}, editors: {\sl
  \bibinfo{booktitle}{Algebra and Coalgebra in Computer Science - 4th
  International Conference, CALCO 2011, Winchester, UK, August 30-September 2,
  2011. Proceedings}}, {\sl \bibinfo{series}{Lecture Notes in Computer
  Science}} \bibinfo{volume}{6859}, \bibinfo{publisher}{Springer}, pp.
  \bibinfo{pages}{378--385}, \doi{10.1007/978-3-642-22944-2\_27}.

\bibitemdeclare{inproceedings}{AcetoCIMR11}
\bibitem{AcetoCIMR11}
\bibinfo{author}{Luca \surnamestart Aceto\surnameend}, \bibinfo{author}{Matteo
  \surnamestart Cimini\surnameend}, \bibinfo{author}{Anna \surnamestart
  Ing{\'o}lfsd{\'o}ttir\surnameend}, \bibinfo{author}{Mohammad~Reza
  \surnamestart Mousavi\surnameend} \& \bibinfo{author}{Michel~A. \surnamestart
  Reniers\surnameend} (\bibinfo{year}{2011}): \emph{\bibinfo{title}{Rule
  Formats for Distributivity}}.
\newblock In \bibinfo{editor}{Adrian~Horia \surnamestart Dediu\surnameend},
  \bibinfo{editor}{Shunsuke \surnamestart Inenaga\surnameend} \&
  \bibinfo{editor}{Carlos \surnamestart Mart\'{\i}n-Vide\surnameend}, editors:
  {\sl \bibinfo{booktitle}{Language and Automata Theory and Applications - 5th
  International Conference, LATA 2011, Tarragona, Spain, May 26--31, 2011.
  Proceedings}}, {\sl \bibinfo{series}{Lecture Notes in Computer Science}}
  \bibinfo{volume}{6638}, \bibinfo{publisher}{Springer}, pp.
  \bibinfo{pages}{80--91}, \doi{10.1007/978-3-642-21254-3\_5}.

\bibitemdeclare{article}{AcetoCIMR11a}
\bibitem{AcetoCIMR11a}
\bibinfo{author}{Luca \surnamestart Aceto\surnameend}, \bibinfo{author}{Matteo
  \surnamestart Cimini\surnameend}, \bibinfo{author}{Anna \surnamestart
  Ing{\'o}lfsd{\'o}ttir\surnameend}, \bibinfo{author}{Mohammad~Reza
  \surnamestart Mousavi\surnameend} \& \bibinfo{author}{Michel~A. \surnamestart
  Reniers\surnameend} (\bibinfo{year}{2011}): \emph{\bibinfo{title}{{SOS} rule
  formats for zero and unit elements}}.
\newblock {\sl \bibinfo{journal}{Theoretical Computer Science}}
  \bibinfo{volume}{412}(\bibinfo{number}{28}), pp. \bibinfo{pages}{3045--3071},
  \doi{10.1016/j.tcs.2011.01.024}.

\bibitemdeclare{inproceedings}{Aceto01}
\bibitem{Aceto01}
\bibinfo{author}{Luca \surnamestart Aceto\surnameend}, \bibinfo{author}{Wan
  \surnamestart Fokkink\surnameend} \& \bibinfo{author}{Chris \surnamestart
  Verhoef\surnameend} (\bibinfo{year}{2001}): \emph{\bibinfo{title}{Structural
  Operational Semantics}}.
\newblock In \bibinfo{editor}{Jan~A. \surnamestart Bergstra\surnameend},
  \bibinfo{editor}{Alban \surnamestart Ponse\surnameend} \&
  \bibinfo{editor}{Scott~A. \surnamestart Smolka\surnameend}, editors: {\sl
  \bibinfo{booktitle}{Handbook of Process Algebra, {Chapter} 3}},
  \bibinfo{publisher}{Elsevier Science, Dordrecht, The Netherlands}, pp.
  \bibinfo{pages}{197--292}, \doi{10.1016/B978-044482830-9/50021-7}.

\bibitemdeclare{inproceedings}{AcetoGI13}
\bibitem{AcetoGI13}
\bibinfo{author}{Luca \surnamestart Aceto\surnameend},
  \bibinfo{author}{Eugen-Ioan \surnamestart Goriac\surnameend} \&
  \bibinfo{author}{Anna \surnamestart Ing{\'o}lfsd{\'o}ttir\surnameend}
  (\bibinfo{year}{2013}): \emph{\bibinfo{title}{SOS Rule Formats for Idempotent
  Terms and Idempotent Unary Operators}}.
\newblock In: {\sl \bibinfo{booktitle}{SOFSEM}}, pp. \bibinfo{pages}{108--120},
  \doi{10.1007/978-3-642-35843-2\_11}.

\bibitemdeclare{inproceedings}{axiom-algebraic}
\bibitem{axiom-algebraic}
\bibinfo{author}{Luca \surnamestart Aceto\surnameend},
  \bibinfo{author}{Eugen-Ioan \surnamestart Goriac\surnameend},
  \bibinfo{author}{Anna \surnamestart Ingolfsdottir\surnameend},
  \bibinfo{author}{Mohammad~Reza \surnamestart Mousavi\surnameend} \&
  \bibinfo{author}{Michel \surnamestart Reniers\surnameend}
  (\bibinfo{year}{2013}): \emph{\bibinfo{title}{Exploiting Algebraic Laws to
  Improve Mechanized Axiomatizations}}.
\newblock In: {\sl \bibinfo{booktitle}{Proceedings of the 5th Conference on
  Algebra and Coalgebra in Computer Science (CALCO 2013)}}, {\sl
  \bibinfo{series}{Lecture Notes in Computer Science}} \bibinfo{volume}{8089},
  \bibinfo{publisher}{Springer-Verlag, Berlin, Germany, 2013}.

\bibitemdeclare{article}{Mousavi09-BEATCS}
\bibitem{Mousavi09-BEATCS}
\bibinfo{author}{Luca \surnamestart Aceto\surnameend}, \bibinfo{author}{Anna
  \surnamestart Ing\'olfsd\'ottir\surnameend}, \bibinfo{author}{Mohammad~Reza
  \surnamestart Mousavi\surnameend} \& \bibinfo{author}{Michel~A. \surnamestart
  Reniers\surnameend} (\bibinfo{year}{2009}): \emph{\bibinfo{title}{Algebraic
  Properties for Free!}}
\newblock {\sl \bibinfo{journal}{Bulletin of the European Association for
  Theoretical Computer Science}} \bibinfo{volume}{99}, pp.
  \bibinfo{pages}{81--104}.

\bibitemdeclare{article}{BaetenV04}
\bibitem{BaetenV04}
\bibinfo{author}{Jos C.~M. \surnamestart Baeten\surnameend} \&
  \bibinfo{author}{Erik~P. \surnamestart de~Vink\surnameend}
  (\bibinfo{year}{2004}): \emph{\bibinfo{title}{Axiomatizing GSOS with
  termination}}.
\newblock {\sl \bibinfo{journal}{J. Log. Algebr. Program.}}
  \bibinfo{volume}{60-61}, pp. \bibinfo{pages}{323--351},
  \doi{10.1016/j.jlap.2004.03.001}.

\bibitemdeclare{article}{Bloom95}
\bibitem{Bloom95}
\bibinfo{author}{Bard \surnamestart Bloom\surnameend}, \bibinfo{author}{Sorin
  \surnamestart Istrail\surnameend} \& \bibinfo{author}{Albert~R. \surnamestart
  Meyer\surnameend} (\bibinfo{year}{1995}): \emph{\bibinfo{title}{Bisimulation
  can't be traced}}.
\newblock {\sl \bibinfo{journal}{J. ACM}} \bibinfo{volume}{42}, pp.
  \bibinfo{pages}{232--268}, \doi{10.1145/200836.200876}.

\bibitemdeclare{incollection}{Brinksma-1985a}
\bibitem{Brinksma-1985a}
\bibinfo{author}{Ed~\surnamestart Brinksma\surnameend} (\bibinfo{year}{1985}):
  \emph{\bibinfo{title}{A Tutorial on \textsc{Lotos}}}.
\newblock In \bibinfo{editor}{Michel \surnamestart Diaz\surnameend}, editor:
  {\sl \bibinfo{booktitle}{Proc.\ Protocol Specification, Testing and
  Verification V}}, \bibinfo{publisher}{North-Holland},
  \bibinfo{address}{Amsterdam, Netherlands}, pp. \bibinfo{pages}{171--194}.

\bibitemdeclare{article}{Brogi98}
\bibitem{Brogi98}
\bibinfo{author}{Antonio \surnamestart Brogi\surnameend} \&
  \bibinfo{author}{Jean-Marie \surnamestart Jacquet\surnameend}
  (\bibinfo{year}{1998}): \emph{\bibinfo{title}{On the Expressiveness of
  Linda-like Concurrent Languages}}.
\newblock {\sl \bibinfo{journal}{Electr. Notes Theor. Comput. Sci.}}
  \bibinfo{volume}{16}(\bibinfo{number}{2}), pp. \bibinfo{pages}{75--96},
  \doi{10.1016/S1571-0661(04)00118-5}.

\bibitemdeclare{article}{MaudeMSOS}
\bibitem{MaudeMSOS}
\bibinfo{author}{Fabricio \surnamestart Chalub\surnameend} \&
  \bibinfo{author}{Christiano \surnamestart Braga\surnameend}
  (\bibinfo{year}{2007}): \emph{\bibinfo{title}{{Maude MSOS Tool}}}.
\newblock {\sl \bibinfo{journal}{Electron. Notes Theor. Comput. Sci.}}
  \bibinfo{volume}{176}, pp. \bibinfo{pages}{133--146},
  \doi{10.1016/j.entcs.2007.06.012}.

\bibitemdeclare{proceedings}{DBLP:conf/maude/2007}
\bibitem{DBLP:conf/maude/2007}
\bibinfo{editor}{Manuel \surnamestart Clavel\surnameend},
  \bibinfo{editor}{Francisco \surnamestart Dur{\'a}n\surnameend},
  \bibinfo{editor}{Steven \surnamestart Eker\surnameend},
  \bibinfo{editor}{Patrick \surnamestart Lincoln\surnameend},
  \bibinfo{editor}{Narciso \surnamestart Mart\'{\i}-Oliet\surnameend},
  \bibinfo{editor}{Jos{\'e} \surnamestart Meseguer\surnameend} \&
  \bibinfo{editor}{Carolyn~L. \surnamestart Talcott\surnameend}, editors
  (\bibinfo{year}{2007}): \emph{\bibinfo{title}{All About Maude - A
  High-Performance Logical Framework, How to Specify, Program and Verify
  Systems in Rewriting Logic}}. {\sl \bibinfo{series}{Lecture Notes in Computer
  Science}} \bibinfo{volume}{4350}, \bibinfo{publisher}{Springer},
  \doi{10.1007/978-3-540-71999-1\_1}.

\bibitemdeclare{inproceedings}{Mousavi08-CONCUR}
\bibitem{Mousavi08-CONCUR}
\bibinfo{author}{Sjoerd \surnamestart Cranen\surnameend},
  \bibinfo{author}{Mohammad~Reza \surnamestart Mousavi\surnameend} \&
  \bibinfo{author}{Michel~A. \surnamestart Reniers\surnameend}
  (\bibinfo{year}{2008}): \emph{\bibinfo{title}{A Rule Format for
  Associativity}}.
\newblock In \bibinfo{editor}{Franck \surnamestart van Breugel\surnameend} \&
  \bibinfo{editor}{Marsha \surnamestart Chechik\surnameend}, editors: {\sl
  \bibinfo{booktitle}{Proceedings of the 19th International Conference on
  Concurrency Theory (CONCUR'08)}}, {\sl \bibinfo{series}{Lecture Notes in
  Computer Science}} \bibinfo{volume}{5201},
  \bibinfo{publisher}{Springer-Verlag}, pp. \bibinfo{pages}{447--461},
  \doi{10.1007/978-3-540-85361-9\_35}.

\bibitemdeclare{inproceedings}{gsos-data}
\bibitem{gsos-data}
\bibinfo{author}{Daniel \surnamestart Gebler\surnameend},
  \bibinfo{author}{Eugen-Ioan \surnamestart Goriac\surnameend} \&
  \bibinfo{author}{Mohammad~Reza \surnamestart Mousavi\surnameend}
  (\bibinfo{year}{2013}): \emph{\bibinfo{title}{Algebraic Meta-Theory of
  Processes with Data}}.
\newblock In: {\sl \bibinfo{booktitle}{Proceedings Tenth Workshop on Structural
  Operational Semantics 2013}}, \bibinfo{series}{Lecture Notes in Computer
  Science}, \bibinfo{publisher}{Springer-Verlag, Berlin, Germany, 2013}.
\newblock \bibinfo{note}{In the current volume}.

\bibitemdeclare{inproceedings}{Glabbeek01}
\bibitem{Glabbeek01}
\bibinfo{author}{R.J. \surnamestart van Glabbeek\surnameend}
  (\bibinfo{year}{2001}): \emph{\bibinfo{title}{The Linear Time - Branching
  Time Spectrum {I}. {T}he Semantics of Concrete, Sequential Processes}}.
\newblock In \bibinfo{editor}{A.~Ponse S.A. Smolka~J.A. \surnamestart
  Bergstra\surnameend}, editor: {\sl \bibinfo{booktitle}{Handbook of Process
  Algebra}}, \bibinfo{publisher}{Elsevier}, pp. \bibinfo{pages}{3--99},
  \doi{10.1007/3-540-57208-2\_6}.

\bibitemdeclare{article}{GoriacCLAG09}
\bibitem{GoriacCLAG09}
\bibinfo{author}{Eugen-Ioan \surnamestart Goriac\surnameend},
  \bibinfo{author}{Georgiana \surnamestart Caltais\surnameend},
  \bibinfo{author}{Dorel \surnamestart Lucanu\surnameend},
  \bibinfo{author}{Oana \surnamestart Andrei\surnameend} \&
  \bibinfo{author}{Gheorghe \surnamestart Grigoras\surnameend}
  (\bibinfo{year}{2009}): \emph{\bibinfo{title}{Patterns for Maude Metalanguage
  Applications}}.
\newblock {\sl \bibinfo{journal}{Electr. Notes Theor. Comput. Sci.}}
  \bibinfo{volume}{238}(\bibinfo{number}{3}), pp. \bibinfo{pages}{121--138},
  \doi{10.1016/j.entcs.2009.05.016}.

\bibitemdeclare{article}{GrooteV92}
\bibitem{GrooteV92}
\bibinfo{author}{Jan~Friso \surnamestart Groote\surnameend} \&
  \bibinfo{author}{Frits~W. \surnamestart Vaandrager\surnameend}
  (\bibinfo{year}{1992}): \emph{\bibinfo{title}{Structured Operational
  Semantics and Bisimulation as a Congruence}}.
\newblock {\sl \bibinfo{journal}{Information and Computation}}
  \bibinfo{volume}{100}(\bibinfo{number}{2}), pp. \bibinfo{pages}{202--260},
  \doi{10.1016/0890-5401(92)90013-6}.

\bibitemdeclare{article}{so55686}
\bibitem{so55686}
\bibinfo{author}{Pieter~H. \surnamestart {Hartel}\surnameend}
  (\bibinfo{year}{1999}): \emph{\bibinfo{title}{{LETOS} - a lightweight
  execution tool for operational semantics}}.
\newblock {\sl \bibinfo{journal}{Software: Practice and Experience}}
  \bibinfo{volume}{29}(\bibinfo{number}{15}), pp. \bibinfo{pages}{1379--1416},
  \doi{10.1002/(SICI)1097-024X(19991225)29:15\%3C1379::AID-SPE286\%3E3.0.CO;2-%
V}.

\bibitemdeclare{article}{DBLP:journals/jacm/HennessyM85}
\bibitem{DBLP:journals/jacm/HennessyM85}
\bibinfo{author}{Matthew \surnamestart Hennessy\surnameend} \&
  \bibinfo{author}{Robin \surnamestart Milner\surnameend}
  (\bibinfo{year}{1985}): \emph{\bibinfo{title}{Algebraic laws for
  nondeterminism and concurrency}}.
\newblock {\sl \bibinfo{journal}{J. ACM}}
  \bibinfo{volume}{32}(\bibinfo{number}{1}), pp. \bibinfo{pages}{137--161},
  \doi{10.1145/2455.2460}.

\bibitemdeclare{book}{Hoare85}
\bibitem{Hoare85}
\bibinfo{author}{C.~A.~R. \surnamestart Hoare\surnameend}
  (\bibinfo{year}{1985}): \emph{\bibinfo{title}{Communicating Sequential
  Processes}}.
\newblock \bibinfo{publisher}{Prentice Hall}.

\bibitemdeclare{article}{DBLP:journals/fmsd/Lin95}
\bibitem{DBLP:journals/fmsd/Lin95}
\bibinfo{author}{Huimin \surnamestart Lin\surnameend} (\bibinfo{year}{1995}):
  \emph{\bibinfo{title}{{PAM}: A Process Algebra Manipulator}}.
\newblock {\sl \bibinfo{journal}{Formal Methods in System Design}}
  \bibinfo{volume}{7}(\bibinfo{number}{3}), pp. \bibinfo{pages}{243--259},
  \doi{10.1007/BF01384078}.

\bibitemdeclare{book}{Milner89a}
\bibitem{Milner89a}
\bibinfo{author}{Robin \surnamestart Milner\surnameend} (\bibinfo{year}{1989}):
  \emph{\bibinfo{title}{Communication and Concurrency}}.
\newblock \bibinfo{publisher}{Prentice Hall}.

\bibitemdeclare{article}{sos-maude}
\bibitem{sos-maude}
\bibinfo{author}{Mohammad~Reza \surnamestart Mousavi\surnameend} \&
  \bibinfo{author}{Michel~A. \surnamestart Reniers\surnameend}
  (\bibinfo{year}{2006}): \emph{\bibinfo{title}{Prototyping {SOS} meta-theory
  in {Maude}}}.
\newblock {\sl \bibinfo{journal}{Electron. Notes Theor. Comput. Sci.}}
  \bibinfo{volume}{156}, pp. \bibinfo{pages}{135--150},
  \doi{10.1016/j.entcs.2005.09.030}.

\bibitemdeclare{article}{Mousavi05-IC}
\bibitem{Mousavi05-IC}
\bibinfo{author}{Mohammad~Reza \surnamestart Mousavi\surnameend},
  \bibinfo{author}{Michel~A. \surnamestart Reniers\surnameend} \&
  \bibinfo{author}{Jan~Friso \surnamestart Groote\surnameend}
  (\bibinfo{year}{2005}): \emph{\bibinfo{title}{Notions of Bisimulation and
  Congruence Formats for {SOS} with Data}}.
\newblock {\sl \bibinfo{journal}{Information and Computation}}
  \bibinfo{volume}{200}(\bibinfo{number}{1}), pp. \bibinfo{pages}{107--147},
  \doi{10.1016/j.ic.2005.03.002}.

\bibitemdeclare{article}{Mousavi05-IPL}
\bibitem{Mousavi05-IPL}
\bibinfo{author}{Mohammad~Reza \surnamestart Mousavi\surnameend},
  \bibinfo{author}{Michel~A. \surnamestart Reniers\surnameend} \&
  \bibinfo{author}{Jan~Friso \surnamestart Groote\surnameend}
  (\bibinfo{year}{2005}): \emph{\bibinfo{title}{A Syntactic Commutativity
  Format for {SOS}}}.
\newblock {\sl \bibinfo{journal}{Information Processing Letters}}
  \bibinfo{volume}{93}, pp. \bibinfo{pages}{217--223},
  \doi{10.1016/j.ipl.2004.11.007}.

\bibitemdeclare{inproceedings}{Park81}
\bibitem{Park81}
\bibinfo{author}{David Michael~Ritchie \surnamestart Park\surnameend}
  (\bibinfo{year}{1981}): \emph{\bibinfo{title}{Concurrency and Automata on
  Infinite Sequences}}.
\newblock In \bibinfo{editor}{Peter \surnamestart Deussen\surnameend}, editor:
  {\sl \bibinfo{booktitle}{Theoretical Computer Science}}, {\sl
  \bibinfo{series}{Lecture Notes in Computer Science}} \bibinfo{volume}{104},
  \bibinfo{publisher}{Springer}, pp. \bibinfo{pages}{167--183},
  \doi{10.1007/BFb0017309}.

\bibitemdeclare{techreport}{Plotkin81}
\bibitem{Plotkin81}
\bibinfo{author}{Gordon~D. \surnamestart Plotkin\surnameend}
  (\bibinfo{year}{1981}): \emph{\bibinfo{title}{A structural approach to
  operational semantics}}.
\newblock \bibinfo{type}{Technical Report} \bibinfo{number}{DAIMI FN-19},
  \bibinfo{institution}{Computer Science Department, Aarhus University},
  \bibinfo{address}{Aarhus, Denmark}.

\bibitemdeclare{book}{SangiorgiW2001}
\bibitem{SangiorgiW2001}
\bibinfo{author}{Davide \surnamestart Sangiorgi\surnameend} \&
  \bibinfo{author}{David \surnamestart Walker\surnameend}
  (\bibinfo{year}{2001}): \emph{\bibinfo{title}{The {$\pi$}-Calculus: A Theory
  of Mobile Processes}}.
\newblock \bibinfo{publisher}{Cambridge University Press},
  \bibinfo{address}{Cambridge}.
\newblock \bibinfo{note}{With a foreword by Robin Milner}.

\bibitemdeclare{article}{SerbanutaRM09}
\bibitem{SerbanutaRM09}
\bibinfo{author}{Traian-Florin \surnamestart Serbanuta\surnameend},
  \bibinfo{author}{Grigore \surnamestart Rosu\surnameend} \&
  \bibinfo{author}{Jos{\'e} \surnamestart Meseguer\surnameend}
  (\bibinfo{year}{2009}): \emph{\bibinfo{title}{A rewriting logic approach to
  operational semantics}}.
\newblock {\sl \bibinfo{journal}{Information and Computation}}
  \bibinfo{volume}{207}(\bibinfo{number}{2}), pp. \bibinfo{pages}{305--340},
  \doi{10.1016/j.ic.2008.03.026}.

\bibitemdeclare{article}{Verdejo2006226}
\bibitem{Verdejo2006226}
\bibinfo{author}{Alberto \surnamestart Verdejo\surnameend} \&
  \bibinfo{author}{Narciso \surnamestart Mart{\'\i}-Oliet\surnameend}
  (\bibinfo{year}{2006}): \emph{\bibinfo{title}{Executable structural
  operational semantics in Maude}}.
\newblock {\sl \bibinfo{journal}{The Journal of Logic and Algebraic
  Programming}} \bibinfo{volume}{67}(\bibinfo{number}{1Ð2}), pp.
  \bibinfo{pages}{226 -- 293}, \doi{10.1016/j.jlap.2005.09.008}.

\end{thebibliography}

\end{document}